\documentclass[runningheads]{llncs}

\usepackage{amsmath}
\usepackage{amssymb}
\usepackage{epsf}
\usepackage{graphics}
\usepackage{epsfig}
\usepackage{latexsym}
\usepackage{enumerate}
\usepackage[normalem]{ulem}
\usepackage{color}
\usepackage{float}
\usepackage{graphicx}
\usepackage[loose]{subfigure}
\usepackage{cite}
\usepackage{algorithm}
\usepackage{algorithmic}
\usepackage{epstopdf}
\usepackage{appendix}
\usepackage{placeins}

\usepackage[pagewise]{lineno}

\newcommand{\edge}[2]{(#1,#2)}
\newcommand{\seppair}[2]{\langle#1,#2\rangle}

\def\O(#1){\ensuremath{\mathcal{O}(#1)}}

\begin{document}

\title{Book Embeddings of $k$-Map Graphs
\thanks{Supported by the Deutsche Forschungsgemeinschaft (DFG), grant
    Br835/20-1.}}
\titlerunning{ }
\author{Franz J. Brandenburg}
\authorrunning{Franz J. Brandenburg}
\institute{University of Passau, 94030 Passau, Germany \\
  \email{brandenb@informatik.uni-passau.de}}

\maketitle

\begin{abstract}
 A \emph{map} is a partition of the sphere into regions that are labeled as
  countries or holes.    A \emph{map graph} has the countries of
a map as its vertices
and there is an edge if and only if the countries are adjacent
  and meet in at least one point. For a $k$-\emph{map graph}, at most $k$ countries
  meet in a point. A graph is $k$-\emph{planar} if it can be drawn in the plane
  with  at most $k$ crossings per edge.

  A $p$-page \emph{book embedding} of a graph is a linear ordering of the
  vertices and an embedding of the edges to $p$ pages, such that there is no
  conflict  in any page, that is any two embedded edges do not twist or cross.
The \emph{book thickness} of a graph
is the minimum number of pages in all book embeddings.

  We show that any $k$-map graph with $n$ vertices admits a book embedding in $6\lfloor
k/2 \rfloor+5$ pages, that can be computed in $O(kn)$ time from its map.
  On the other hand, there are $k$-map graphs that
  need $\lfloor 3k/4 \rfloor$  pages.
  In passing, we obtain an improved upper bound of eleven pages for 1-planar
  graphs   and of 17 pages for optimal 2-planar graphs.
\end{abstract}

\section{Introduction} \label{sec:intro}

A  $p$-page \emph{book embedding}  of a graph consists of a
\emph{linear ordering} of the vertices, which is defined by placing
them from left to right, and an \emph{embedding}  of the edges in
$p$ pages, such that there is no  conflict  in any
page. For two vertices $u$ and $v$, let $u<v$ if $u$ precedes $v$ in
the linear ordering and let $u\leq v$ if $u<v$ or $u=v$.
If $u \leq x$, then two edges $\edge{u}{v}$ and $\edge{x}{y}$
\emph{twist}  or \emph{cross}   if
$u < x < v < y$. They \emph{nest} if $u \leq x < y \leq v$ and are
\emph{disjoint} if $u < v \leq x < y$. There is a \emph{conflict}
in a page if any two edges twist that are embedded in the page.
For sets of vertices $U$ and $W$ let $U<W$ if $u<w$ for all $u \in U$
and $w \in W$.
 An \emph{interval} $[u,w]$
consists of all vertices $v$ with $u \leq v \leq w$. Vertex $v$  is
\emph{outside} the interval if  $v \leq u$ or $v \geq w$. Thus the
vertices on the boundary are both  in and outside. Obviously, two edges
do not twist if there is an interval such that both vertices of one
of them are in  and the vertices of the other edge are outside the
interval. If $U$  is a set of vertices, then let $[U]$ be th interval
that contains exactly the vertices of $U$. For $w < U$ let $[w,U]$
denote the interval   $w \leq v \leq u$ for $u \in U$.
An interval $[U,w]$ is defined accordingly.

The \emph{book thickness} of a graph $G$ is the minimum number of
pages in  all book embeddings of $G$. Book thickness is also known as
  stacknumber or pagenumber \cite{dw-llg-04,hls-cqsmlg-92}.
It has been  been studied in areas such as  Graph
Theory, Graph Algorithms, and Graph Drawing.   The book thickness
of $n$-vertex graphs with $m$ edges is at most $\sqrt{m}$
\cite{m-edgepage-94} and at least $\lceil \frac{m-n}{n-3} \rceil$
\cite{bk-btg-79}.
The complete graph $K_n$ has book
thickness $\lceil n/2 \rceil$ \cite{bk-btg-79}.
 Every nondiscrete  outerplanar graph has
 book thickness one. A graph has book thickness at most two
  if and only if it is a subgraph of a planar graph with a Hamiltonian
cycle \cite{bk-btg-79}.  Every  planar graph  has book thickness at
most four. The upper bound has been shown  by Yannakakis
\cite{y-epg4p-89} by a linear time algorithm that
 constructs a 4-page book embedding of a planar graph. Recently,
  Bekos et al.~\cite{kbkpru-4pages-20} and  Yannakakis \cite{y-4pages-20}
 have  shown that some planar graphs need four pages, so that the bound is tight.

There are several approaches to extend the planar graphs, for
example by drawings on surfaces of higher genus \cite{gt-tgt-87},
forbidden minors \cite{d-gt-00},   drawings in the plane with
restrictions on crossings \cite{dlm-survey-beyond-19}, or
generalized adjacency relations \cite{cgp-mg-02}.
   Graphs with bounded genus have constant book thickness \cite{m-genuspage-94}.
Also minor-closed graphs, e.g., graphs with constant tree-width,
have constant book thickness \cite{dw-gtgtp-07}.
 A graph is $(g, k)$-\emph{planar}  if it can be drawn on a surface of Euler genus at
most $g$ with at most $k$ crossings per edge
  \cite{df-stackqueue-18}. Clearly, $(0, 0)$-planar graphs are
the planar graphs  and $(0, k$)-planar graphs are known as
$k$-\emph{planar graphs} \cite{pt-gdfce-97}. An $n$-vertex
$(g,k)$-planar graph with fixed $g$ and $k$ has book thickness
$O(\log n)$  
\cite{df-stackqueue-18}.
 For $k$-planar graphs, this improves the $O(\sqrt{n})$ bound from
\cite{m-edgepage-94} to $O(\log n)$.    Bekos~et
al.~\cite{bbkr-book1p-17} have shown that the book thickness of 1-planar
graphs is constant and that 39 pages suffice.

Recently, Bekos et al.~\cite{bdggmr-benpsf-20,bdggmr-benpsf-20-a, bdggmr-benpsf-22-b}
have introduced $k$-framed graphs that consist of
a planar graph with faces of degree at most $k$, such that
there are crossed edges in the interior of each faces. They
allow   crossed multi-edges but no multi-edges in the frame. The
latter admit smaller faces, for example for 1-planar graphs. A
framed multigraph is \emph{maximal} if every face $f$ of degree $k$
induces a $k$-clique by the edges in the boundary and the crossed
edges in the interior of $f$.
Bekos et al.~state
  a bound of $6\lceil k/2 \rceil+5$ pages for any $k$-framed
graph, which shall be raised to $6\lceil k/2 \rceil+7$ to correct an error
in the earlier versions\footnote{personal communication by S. Griesbach}.

A    $k$-\emph{map}
is a partition of the sphere  into disc homeomorph regions that are labeled as countries or
holes, such that at most $k$ countries meet in a point.
It generalizes planar duality by holes and  an
adjacency in a point. The latter admits large cliques.   A $k$-map
 defines a $k$-\emph{map graph}  with the countries as vertices
and an edge if and only if two  countries meet in at least one
point.   Map graphs have been introduced by Chen et al.~\cite{cgp-mg-02},
who  have shown that any $n$-vertex $d$-map graph has at most $k(n-2)$
edges \cite{c-edges-kmap-04} and admits a clique  of size $\lfloor 3k/2
\rfloor$. If $k$ is small, then $k$-map graphs are related to
$k$-planar graphs. Chen et al.~\cite{cgp-mg-02} have observed that the 2-
and 3-map graphs are the planar graphs.  The 4-map graphs are the
kite-augmented 1-planar graphs \cite{b-4mapGraphs-19}  and the 5-map
graphs are the clique-augmented 2-fan-crossing graphs
\cite{b-5maps-19}. A graph is kite-augmented 1-planar if it has a
drawing such that every edge is crossed at most once and a pair of
crossed edges induces a $K_4$. A graph is clique-augmented
2-fan-crossing if every edge is crossed at most twice. Moreover, if
an edge is crossed by two edges, then the crossing edges are
incident to a common vertex, such that there is a $K_5$ induced by the
vertices of the edges involved in the crossings.
Bekos et al.~\cite{bdggmr-benpsf-20,bdggmr-benpsf-20-a, bdggmr-benpsf-22-b}
have shown that any $k$-map graph is a subgraph of a $2k$-framed graph.
Hence, one  obtains an upper bound of $6k+7$ for the book thickness of $k$-map
graphs by their approach.
  Map graphs are simple, but their representations
allow  multi-edges.  We show that a graph is the simplification of a
maximal $k$-framed multigraph if and only if it is a $k$-map graph.
Hence, $k$-framed multigraphs and $k$-map graphs have the same book
thickness, since multi-edges don't matter for book embeddings.\\

\noindent \textbf{Our contribution.}
We establish improved upper an lower bounds on the book thickness of $k$-map graphs.
The lower bound of  $\lceil k/2
\rceil$ is raised to $\lfloor 3k/4 \rfloor$ using larger cliques.
For the upper bound, we first show that any $k$-map graph is a
$k$-framed multigraph consisting of a planar multigraph with faces of
degree at most $k$. Then we use a modification of Yannakakis algorithm
for the embedding of planar graphs.
 We introduce block-expansions as a new method for the computation of the vertex
ordering. The  embedding of edges is done in two phases.
For any  2-level framed multigraph, first,
only the edges of the outer cycle, the inner edges, and
all edges that are incident to the first outer vertex of a face are
embedded in three pages. The remaining edges of any face of degree $k$ are embedded in
a set of at  most $ \lfloor \frac{k}{2} \rfloor$ pages.
There is an outerplanar and consistent  face-conflict graph, such that two remaining edges
from any two faces do not twist if the  faces have the same color. As outerplanar
graphs are 3-colorable, any
  2-level $k$-framed multigraph  can be embedded in $3 \lfloor
k/2 \rfloor +3$ pages. Twice this number suffices for $k$-framed
multigraphs, where one page can be saved as in Yannakakis
5-page algorithm \cite{y-epg4p-89}, so that we obtain an upper bound
of  $6 \lfloor k/2 \rfloor +5$.
  In addition, we show that the
book thickness of 1-planar graphs is at most eleven.

In the remainder of this paper, we introduce basic notions in
Section~\ref{sec:prelim} and establish the relationship between
framed multigraphs and map graphs. The book embedding of 2-level
framed multigraphs is described in Section~\ref{sec:2-level}. We
study the  composition and applications  in
Section~\ref{sect:application} and state some open problems in
Section~\ref{sec:conclusion}.

\section{Preliminaries}  \label{sec:prelim}
We consider  undirected multigraphs $G = (V,E)$ with sets of $n \geq
1$ vertices $V$ and edges $E$, some of which have multiple copies.
Self-loops are excluded.
An undirected edge is denoted by a pair $\edge{u}{v}$, since it is
later oriented from $u$ to $v$. A graph is $k$-\emph{planar} ($k
\geq 0$) if it admits a drawing (or a topological embedding) such that
every edge is crossed at most $k$ times. Clearly, 0-planar graphs
are \emph{planar}. A planar drawing has \emph{faces}, which are
hole-free regions if the planar graph is connected. A face
is specified by the set $V(f)$  of \emph{vertices} in its boundary.
 If $V(f)$ contains $k \geq 2$ vertices, then $f$ has
\emph{degree} $k$ and is called a $k$-\emph{face}.
For convenience, we use (mixed) sets of vertices and edges
to specify a face, e.g., a vertex $a$ and an edge $\edge{b}{c}$ for
a triangle $(a,b,c)$.

 A \emph{multi-edge} between two vertices consists of several copies of
an edge, one of which is the \emph{original} or 0-\emph{copy}. All
but one copy  is removed if  a multigraph is \emph{simplified}.
A simplified multigraph is a (topologically) \emph{simple} graph.
Multi-edges shall be \emph{non-redundant}, such that there is no
2-face with two copies of an edge as its boundary. Hence, there are
vertices in the interior and the exterior of a 2-\emph{cycle} formed
by  two copies of an edge.
For convenience, we shall not distinguish between a planar
multigraph and a planar drawing (embedding), such that we speak of
vertices, edges, and faces of a multigraph.

\begin{figure}[ht]
  \centering
  \subfigure[]{
    \includegraphics[scale=0.8]{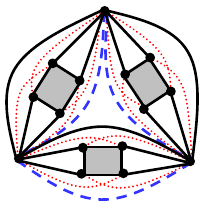}  
    \label{fig:Wconf}
    }
    \hspace{5mm}
 \subfigure[]{
    \includegraphics[scale=0.8]{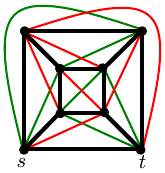}
    \label{fig:crossed-cube}
    }
\hspace{5mm}
    \subfigure[]{
    \includegraphics[scale=0.5]{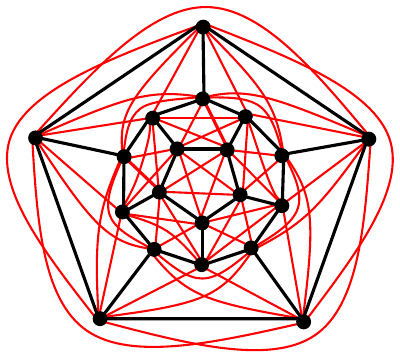}
    \label{fig:Xdodecaeder}
    }
\caption{(a) A 4-framed multigraph, that is a 1-planar graph
consisting of three W-configurations. The graph is a 7-framed graph
if two copies of edges (drawn blue and dashed) and the crossed edges
are removed from the inner
 face.
(b) The 1-planar crossed cube. (c) The 5-framed crossed dodecahedron
graph.
  }
  \label{fig:1planar}
\end{figure}

There is a close relationship between $k$-map graphs and $k$-framed
graphs, that consist of a planar frame with faces of degree at most $k$
and of a set of crossed edges in the interior of each face.
Bekos et al.~\cite{bdggmr-benpsf-20,bdggmr-benpsf-20-a, bdggmr-benpsf-22-b}.
  have shown that any $k$-map graph is a subgraph of a maximal $2k$-framed graph.
We show that multi-edges help to obtain smaller faces.
Multiple adjacencies between countries are
natural for $k$-maps.
Chen et al.~\cite{cgp-mg-02} have shown that any $k$-map graph
admits a representation by a planar graph
\cite{cgp-mg-02}. Let $W=(V, P, L)$ be a planar bipartite graph,
whose first set  is in one-to-one correspondence to
the set of vertices of a graph $G$.
Each vertex of the second set $P$ is called a \emph{point},
that is used to establish edges. Set $L$ consists of 2-sets $\{v,p\}$ with
$v \in V$ and $p \in P$, called a \emph{link}.
The \emph{half-square} of  $W$ is a graph $G = H^2(W)$
with vertex set $V$
such  that there is an edge $\edge{u}{v}$ in $G$ if and only if
there is a point $p$ and links $\edge{u}{p}$ and $\edge{v}{p}$ in
$W$. Then $W$ is called a \emph{witness} of $G$.
Graph $W$ is a $k$-witness if any point has degree at most $k$.
In particular, there are points of degree two, called \emph{2-points}.
If $p$ is a 2-point of $W$ with links $\edge{u}{p}$ and $\edge{v}{p}$,
then $p$ subdivides edge $\edge{u}{v}$ of $G$.
 Conversely, there is an
edge contraction at a 2-point in the half-square.

Chen et al.~\cite{cgp-mg-02} have shown that there
 is a one-to-one correspondence
between  points in maps and witnesses, that leads to the
following characterization.

\begin{proposition}
A graph $G$ is a $k$-map graph if and only if
$G=H^2(W)$ for a planar $k$-witness $W$
\end{proposition}

We wish to normalize any $k$-map graph, similar to a triangulation of a planar graph.
It is often easier to work with normalized graphs than with general ones.
We do so by adding multi-edges that shall be uncrossed in a drawing.
Note that an edge $\edge{u}{v}$ of a $k$-map graph is uncrossed
if the countries for $u$ and $v$ in a map meet in a segment and not just in a point.

 A witness $W$ is called \emph{planar-maximal}
if (1) each face  in  a planar drawing of $W$  is a  quadrangle or a hexagon, (2)
there is a 2-point $p'$ with links $\edge{u}{p'}$ and $\edge{p'}{v}$
if $p$ is a point of degree   $d \geq 3$, such that $u$ and $v$ are
consecutive neighbors at $p$, and (3)
 there are no quadrangular faces with two vertices and two
2-points in the boundary.  Note that the added 2-points   are redundant in the
sense of \cite{cgp-mg-02}, since they define edges that are defined
by $p$. The \emph{planar skeleton} $\mathcal{P}(W)$ is the subgraph
of a planar-maximal witness, in which all $d$-points for $d\geq 3$ are
removed.
A $k$-map graph $G$ is \emph{planar-maximal} if $G=H^2(W)$ for a
planar-maximal witness $W$ with $d$-points for $d \leq k$. Its
\emph{planar skeleton} $\mathcal{P}(G)$ is $H^2(\mathcal{P}(W))$,
which is a planar multi-graph with multiple copies of an edge
$\edge{u}{v}$, one for each 2-\emph{path} $(u,p,v)$ in the  planar
skeleton of $W$ consisting of a 2-point and two links. This is
relevant for the definition of faces if there are separation pairs,
and it extends $k$-framed graphs. Multiple copies of an edge are
ignored for the book embedding. The restriction to quadrangles and
hexagons implies that $G$ is a hole-free map graph, that is, graph $G$ admits
 a map without holes \cite{cgp-rh4mg-06}.
Then $G$ is 2-connected \cite{cgp-rh4mg-06}.
A point of a witness is \emph{redundant}  if all pairs of its
neighbors can also be connected through other points
\cite{cgp-mg-02}. In particular,  there may be many 2-points
connecting two vertices $u$ and $v$. Then the size of the set of
points of a witness is no longer related to the size of its set of
vertices. This resembles the situation of multi-edges in (planar)
graphs. We avoid this situation by the exclusion of \emph{duplicate}
2-points. Two 2-points are duplicates if there is a quadrangular
face with two 2-points and two vertices. Chen et
al.~\cite{cgp-mg-02} have shown that a witness without redundant
points has at most $3n-6$ points and  $O(kn)$ edges, which was
improved to $kn-2k$ \cite{c-edges-kmap-04}, if it has $n$ vertices
and points of degree at most $k$. Hence there are $O(kn)$ 2-points
if duplicates are excluded. Thus we can assume that a planar-maximal
witness for an $n$-vertex $k$-map graph has $O(kn)$ points,
 and that the (planar-maximal) half-square has $O(kn)$ edges. Note
 that $k = (n-1)/2$ if $G$ is an $n$-clique.

There is a \emph{normal form} for planar-maximal $k$-map graphs,
which generalizes the normal form for 1-planar graphs by Alam et
al.~\cite{abk-sld3c-13}. It is obtained via its witness, that is
augmented similar  to a triangulation of a planar graph.

\begin{lemma} \label{lem:normalform}
(i) If $G$ is a planar-maximal $k$-map graph, then there
 is a    planar graph $G'$ with multi-edges and $d$-faces with
$d \leq k$,  such that $G$ is obtained by expanding each $d$-face to
  a $d$-clique, and then removing multi-edges.

(ii) For any $k$-map graph $G=(V,E)$ there is a planar-maximal $k$-map
supergraph $G'=(V', E')$  with $V' =V$ and $E \subseteq E'$,  that can be constructed in
linear time.
\end{lemma}

\begin{proof}
(i) There is a planar-maximal witness $W$ such that  $G = H^2(W)$ and
$G'=H^2(\mathcal{P}(W))$ is a planar multi-graph, which is
2-connected, since the faces of $W$, including the outer face,  are
quadrangles or hexagons \cite{cgp-rh4mg-06}. Graph $\mathcal{P}(W)$
has only 2-points. There is a one-to-one correspondence between
$d$-faces of $G'$ and $2d$-faces of $\mathcal{P}(W)$, since
$\mathcal{P}(W)$ does not admit quadrangular faces with two
2-points. For any $2d$-face of $\mathcal{P}(W)$ there is a
$d$-point that is adjacent to the vertices of the face. Hence, it
creates a $d$-clique in the $d$-face of $G'$. By assumption, $k$-map
graphs are simple, such that there are no multi-edges.

(ii) Assume that $W$ is a witness without redundant points. Then
it has $O(n)$ points \cite{cgp-mg-02}. A planar-maximal augmentation
$W^+$ of a witness $W$   can be constructed in linear  time in the
size of the set of vertices of $W$ if there are no  duplicate
2-points.
 For the construction of $W^+$, first partition
the $d$-faces of $W$ with $d \geq 8$ by 2-paths such that only
$d'$-faces with $d' \leq 6$ remain. This operation creates new links
for the half-square, including new edges and thereby multiple copies of an
edge for $G$. It
  generalizes the triangulation of planar graphs and the
augmentation of   1-planar graphs to    kite-augmented ones
\cite{b-4mapGraphs-19}. Then add a 2-point $p'$ and links
$\edge{u}{p'}$ and $\edge{p'}{v}$ if vertices $u$ and $v$ are
consecutive neighbors at point $p$ if $p$ has degree at least three.
Links $\edge{u}{p'}$ and $\edge{p'}{v}$ can be routed close to the
2-path $(u,p,v)$, such that they are uncrossed. This creates
multi-edges in the half-square if there are different routes for the
2-paths. Finally, remove duplicate 2-points by merging  2-points in
quadrangles with two vertices. Clearly, $H^2(W^+)$
is a supergraph  of $H^2(W)=G$. Clearly, all taken steps can be done in linear
time in the size $n$ of the set of vertices of $W$ if $W$ has no
duplicate 2-points.
  \qed
\end{proof}

A $k$-\emph{framed multigraph} $G$ consists of a \emph{frame} $F(G)$
and of sets of  \emph{crossed edges}. The frame is a spanning planar
subgraph of $G$ with nonredundant multi-edges. A face of  the frame is a
$d$-face with $3 \leq d \leq k$ that  contains a set of
\emph{crossed edges}.
The crossed edges are drawn in the interior of
the face, see Figure~\ref{fig:1planar}. There are also crossed
multi-edges if there are  copies in different faces. The \emph{set of edges}
$E(f)$ of face $f$ consists of the edges in the boundary of $f$ and
the crossed edges in its interior.
Any face has a distinguished vertex,
called its \emph{first outer vertex}, that is denoted by $\alpha(f)$.
So $E(f)$ is partitioned into the set
$E_{\alpha}(f)$ of edges incident to the first outer vertex and the
remainder $E^-(f)$. If $f$ is a $d$-face and $E(f)$ induces a
$d$-clique, then $E^-(f)$ induces a $(d-1)$-clique.

We assume that framed multigraphs are biconnected, since the book
thickness of a graph is the maximum book thickness of its
biconnected components \cite{bk-btg-79}. In addition, we assume that
the frame is  biconnected, which is useful later on.

\begin{lemma} \label{lem:2-connected}
For any biconnected $k$-framed multigraph $G$ there is a
$k$-framed multigraph $G'$ on the same set of vertices and with the
same set of crossed edges in each face, such that the frame of $G'$
is biconnected and is a planar supergraph of the frame of $G$.
\end{lemma}
\begin{proof}
If the frame of $G$ is disconnected, then a face $f$ has a hole with
an inner component $M$. There are crossed edges between vertices in
the boundary of $f$, in the outer face of $M$, and between vertices
of $f$ and $M$. If $M$ consists of a single vertex, then connect it
to an edge in the boundary of $f$ such that there is a triangle. Otherwise, consider
an edge $e$ in the outer face of $M$ and an edge $e'$ in the
boundary of $f$. Create an internally triangulated quadrangle with
$e$ and $e'$ on opposite sides. Then $M$ is biconnected to the
component with face $f$. Every crossed edge can be routed in the
interior of the new face $f'$, whose boundary consists of the
boundary of $f$ and the outer face of $M$.
 Finally, create a triangle with $v$, its predecessor in one component
and its successor in the other component if  there is a cutvertex
$v$ in the frame of $G$.
  \qed
\end{proof}

 For convenience, we assume that framed multigraphs are \emph{maximal}
such that any  face of degree $d$ induces a $d$-clique. Then there
may be crossed multi-edges  in the interior of  faces  that may not
be adjacent. Clearly, the drawing of a framed multigraph can be
augmented to a maximal one, first by establishing  2-connectivity of
the frame and then  by filling the interior of each face such that
there is a clique. In addition, we assume that the outer face is a
triangle (or that there are no crossed edges in the outer face),
which is obtained as before when establishing  2-connectivity in
Lemma~\ref{lem:2-connected}. However, there are no crossed edges
incident to the vertices of the outer triangle.

 A \emph{separation pair} $\seppair{s}{t}$ of graph $G$
is such that  $G-\{s,t\}$ partitions
into at least two connected components. It is an \emph{inner
separation pair} if vertex $t$ (or $s$) is not in the outer face of a given
drawing of $G$. A component without vertices in the outer face is
called an \emph{inner component}. In general, there are several inner
components that share exactly vertices $s$ and $t$.

 Note that the book thickness of a graph
is bounded by the book thickness of any augmentation by vertices and
edges. Hence, we consider maximal $d$-framed multigraphs for our
study of an upper bound on the book thickness.\\

There is a close relationship between  framed multigraphs  and  map
graphs.

\begin{theorem} \label{thm:map-multiframe}
Any $k$-map graph is the simplification of a maximal $k$-framed
multigraph.
\end{theorem}
\begin{proof}
Chen et al.~\cite{cgp-mg-02} have shown that a graph $G$ is a $k$-map graph
if and only if it is the half-square of a $k$-witness $W$ such that
$G=H^2(W)$.
 A witness admits the construction of a frame as
follows. Consider a planar drawing of $W$.
For any point  $p$ of $W$,  add  a cycle of 2-paths
around $p$. A 2-path consist of a point $t$ of degree two and edges
$\edge{u}{t}$ and $\edge{t}{v}$ for vertices $u$ and $v$ that
  are consecutive at $p$. There is a multi-edge between $u$ and $v$ if
  there are 2-paths around several points. We assume that there is no face
in a drawing of the augmentation of $W$ containing two 2-points (and
two vertices), similar to nonredundant multi-edges. By the
half-square there is an uncrossed edge in $G$ from every 2-path in
$W$. Hence, every $k$-point for $k \geq 3$ is in a face formed by
the 2-paths of its neighbors, which is a $k$-face in $F(G)$. It
defines a $k$-clique. Every other edge of $G$ is a crossed edge in
the interior of a face of $F(G)$, that is the edge is created by the
2-path between two neighbors of a point.
  \qed
\end{proof}

\begin{lemma} \label{lem:multiframe-map}
The simplification of any maximal $k$-framed multigraph is a $k$-map
graph.
\end{lemma}
\begin{proof}
Construct a $k$-witness $W$ from a maximal $k$-framed multigraph $G$
as follows.
First, subdivide every edge of the frame by a point, which it taken
as a 2-point of $W$.  Then add a
$k$-point in each $k$-face and connect it to the vertices in the
boundary. This creates a $k$-clique for each face of degree $k$,
which is feasible, since $G$ is maximal.  Clearly, any edge of the frame
is represented in $W$ by the 2-path with the added 2-points, and conversely,
and any crossed edge in the interior of a face is represented via the inserted
$k$-point, and conversely. As  $k$-map graphs are simple, that is $H^2(W)$,
we must simplify the given maximal $k$-framed multigraph.
\qed
\end{proof}

\begin{corollary}  \label{cor:map=frame}
 Any simple subgraph of a  $k$-framed multigraph  is a subgraph of a $k$-map graph.
\end{corollary}

Note that a subgraph of a $k$-map graph is not necessarily a $k$-map
graph. In fact, the removal of an edge  from a 4-map graph may
result in a non 4-map graph, as shown  by Chen et
al.~\cite{cgp-mg-02}. This fact is due to the need for an
augmentation, such as kite-augmented 1-planar \cite{b-4mapGraphs-19}
and clique-augmented 5-planar graphs \cite{b-5maps-19}.

\section{Two-Level Graphs} \label{sec:2-level}

We recall basic notions from \cite{y-epg4p-89} and extend them for our needs.
Familiarity with Yannakakis approach for 2-level planar graphs will be helpful.
Basically, we traverse distinguished sets of blocks by Yannakakis nested method
and treat  them as a single X-block.

 The \emph{peeling technique}, introduced by Heath~\cite{heath1984embedding},
has been used in all later approaches on upper bounds for the book
thickness of generalized planar graphs \cite{ 
bbkr-book1p-17,bdggmr-benpsf-20, df-stackqueue-18,y-epg4p-89}.
 It decomposes a graph into 2-level graphs and computes a \emph{leveling}  of the
vertices of a graph, such that there are layered separators
\cite{d-glls-15}. 
So the computation of a book embedding of a graph is reduced to that
of its 2-level subgraphs. The peeling technique generalizes
canonically to planar multigraphs.

The vertices in the outer face of a planar drawing are at level
zero.  Vertices are in level $\ell+1$ if they are in the outer face,
when all vertices at levels at most $\ell$ 
are removed. So there are no edges between vertices in levels $i$ and
$j$ if $|i-j| >1$ if the peeling technique is used, both for a planar
multi-graph and a $k$-map   or $k$-framed multigraph. In consequence,
the book embedding of such  graphs can be composed of the book embedding
of its 2-level subgraphs at odd and even levels, so that the book
thickness  of a  graph  is at most twice the book
thickness of its 2-level subgraphs.   

\subsection{Planar 2-Level Multigraphs}

\begin{figure}[t]
\centering
    \includegraphics[scale=0.8] {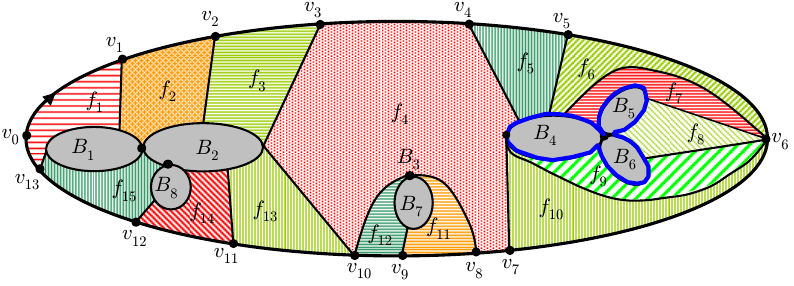}   
  \caption{
  A planar 2-level multigraph with blocks $B_1, \ldots, B_{8}$
    including the  elementary block $B_3$ as the leader of $B_7$.
    Blocks $B_5$ and $B_6$ are covered by $v_6$, which is the last outer
vertex of the block-tree consisting of $B_4, B_5$ and $B_6$. Now $B_4$
is expanded to $B^*_4$   that includes $B_4, B_5$ and $B_6$.
 There are faces $f_1, \ldots, f_{15}$.
  It is assumed that more vertices are placed on the outer cycle and in the boundary
  of the blocks, and that each face contains a set of crossed edges.
 Bad  faces are drawn red   and good ones green.
  For example, $f_1$ is bad for $B_1$  and is in conflict with
  $f_2$ and $f_{15}$,
  $f_4$ is bad for $B_3$ and $B_4$ and is in conflict with  $f_5, f_6, f_{10}, f_{11}$
 and $  f_{12}$,
and $f_7$ is bad for $B_5$  and is in conflict with
  $f_8$ and $f_{9}$.
    The vertex
ordering is  $v_0, B_1,  v_1, B_2, v_2, v_3,  B_3,  B'_4, B_6, B_5,
B''_4, v_4, v_5, v_6, v_7, v_8, B_7, v_9, v_{10}, v_{11}, B_{8}$,
$v_{12}, v_{13}$ with $B_4=B'_4\, b_4 \, B_4''$, where $b_4$ is the leader of
  $B_5$ and $B_6$.
  }
  \label{fig:2-level-graph}
\end{figure}

A planar 2-\emph{level (multi-) graph}   is the subgraph  induced by
a cycle $C$ of level $\ell$ vertices, called \emph{outer vertices},
and of the level $\ell+1$ vertices in its interior, called
\emph{inner vertices}, see Figure~\ref{fig:2-level-graph}. The
subgraph $I$ in the interior  is composed of \emph{blocks}. A block
is the cycle of outer vertices of a 2-connected component. It will be
the outer cycle at the next level.
It may consist of an
edge with its vertices or of a single vertex, which is called an
 \emph{elementary block}. Two blocks may
share a vertex, which is a \emph{cutvertex} of $I$. These vertices
are distinguished as the leader of blocks. A connected component of
$I$ is called a \emph{block-tree}. It is a cactus consisting of
blocks with branches at cutvertices.
  Two block trees are separated by  chords between outer vertices or a face
  that can contain  such a chord. For example, the
  frame of the graph in Figure~\ref{fig:Wconf} has three
  block trees, each consisting of a single quadrangle.
 By  Lemma~\ref{lem:2-connected}, we can assume that planar 2-level multigraphs are
  biconnected. \\

  Yannakakis \cite{y-epg4p-89} has simplified the problem of embedding a
  planar 2-level graph into a 3-page book by the assumption that the
  graph is triangulated and that the inner subgraph is connected.
 Connectivity  can be achieved at the expense of planarity.
There are outer chords if a planar 2-level graph is triangulated and
the inner subgraph is not connected. Then two connected components
can be connected by an additional edge, which crosses the outer chords
that separate them.

For our book embedding of planar 2-level multigraphs, we follow the
block oriented description by Yannakakis \cite{y-epg4p-89}, see also
  \cite{
bbkr-book1p-17}. The one by Bekos et
al.~\cite{bdggmr-benpsf-20} can be regarded as face oriented. The
edges of a planar 2-level multigraph are \emph{outer edges} on the
outer cycle $C$, \emph{outer chords} between non-consecutive outer
vertices in the interior of $C$, \emph{binding edges} between inner
and outer vertices (in this direction), that are classified into
\emph{forward} and \emph{backward binding}, and \emph{inner edges}
between two vertices that are consecutive for a block. There are no
\emph{inner chords} between two non-consecutive inner vertices and
no copies of inner edges. Such edges are flipped into the interior
of a block and are considered at the next level.

The \emph{faces} of a planar 2-level multigraph are the faces
between $C$ and $I$ in the interior of $C$. The outer face and the
faces in the interior of   blocks are discarded. Each face contains
at least one outer vertex. It may  contain one or two outer chords.
Faces may contain vertices and edges from
blocks in several block-trees. A face has $r$ binding edges, where
$r\geq 0$ is even by the alternation between outer and inner
vertices.

Each block $B$ has a least vertex $\lambda(B)$, called the
\emph{leader}, which is the cutvertex of $B$ and its parent in $I$
if $B$ is not the root  in a block tree.  For
$B=b_0,\ldots, b_q$ with $q \geq 0$  let $b_0=\lambda(B)$ and traverse $B$
in ccw-order.
A cutvertex may be the leader of several blocks,
that are 
ordered clockwise like the outer cycle. The leader of a block plays
a special role, see also \cite{y-epg4p-89}.
 Any inner vertex is in a single block, except if
it is a cutvertex or the first vertex of a block tree.
 The \emph{root}  of a block tree is special.
A block tree $\mathcal{T}$ has a \emph{first face}
$f_{\mathcal{T}}$, which is the least face containing a vertex of
$\mathcal{T}$. The least outer vertex in this face is denoted
by $\alpha(\mathcal{T})$, and is also called the first outer vertex of $\mathcal{T}$.
By 2-connectivity, $f_{\mathcal{T}}$
has a \emph{last binding edge} $\edge{a_0}{v_s}$ between a vertex $a_0$ of
any  block
of  $\mathcal{T}$ and an outer vertex $v_s$. Vertex $v_s$ is
searched by a ccw-traversal of $f_{\mathcal{T}}$ from its least
outer vertex.   Vertex $a_0$ is set to be the leader of the root
 of $\mathcal{T}$ and is called the \emph{first vertex} of
$\mathcal{T}$, denoted $\lambda(\mathcal{T})$. Vertex
$v_s$ is called  the \emph{last outer vertex} of $\mathcal{T}$,
denoted $\omega(\mathcal{T})$.
Hence, the first outer face $f_{\mathcal{T}}$ of any block-tree $\mathcal{T}$
contains the vertices $\alpha(\mathcal{T}), \lambda(\mathcal{T})$
and $\omega(\mathcal{T})$. Observe that a face may contain
the root of several block-trees and vertices   $\alpha(\mathcal{T}_i)$
and $\lambda(\mathcal{T}_i)$ for $i=1,\ldots, r$ and $r \geq 0$,
that all have the same first outer vertex. However,
any face contains the first vertex $\lambda(\mathcal{T})$ of at most
one block-tree, since there is the edge
$\edge{\lambda(\mathcal{T})}{\omega(\mathcal{T})}$ and another binding edge
between a vertex of $\mathcal{T}$ and an outer vertex by biconnectivity,
so that there is a closed curve through $\omega(\mathcal{T})$
and $\lambda(\mathcal{T})$, that separated $\mathcal{T}$ from the remainder of $H$.
A face may contain the first and last outer vertex of several block-trees, since
we allow multi-edges.

\begin{figure}[t]
\centering
\subfigure[ ] {    
    \includegraphics[scale=0.8]{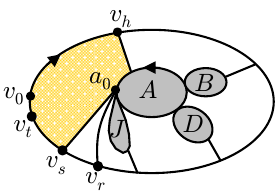}  
      \label{fig:root-a}
  }
  \hspace{2mm}
\subfigure[ ] {    
    \includegraphics[scale=0.8]{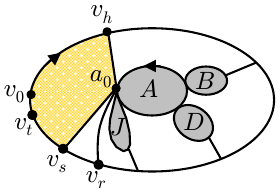}
      \label{fig:root-b}
  }
 \caption{
  The first vertex $a_0$ and the
   root  of a block tree $\mathcal{T}$ with blocks $A, B, D, J$
   in this order and last outer vertex $v_r$.
   (a) Block $A$ including $a_0$ is dominated by $v_0$.
   (b) There is an elementary block $a_0$ with dominator  $v_0$.
  Block $A$ is dominated by $v_h$. Vertex $a_0$ is the leader of $A$.
  }
  \label{fig:firstblock}
\end{figure}

Any inner vertex is in a single block except if its is a cutvertex or the
first vertex of a block-free. For uniqueness,
 we assign each vertex to the block that is closest to the
root in its block tree, and we denote the set of vertices
 assigned to $B$ by $V(B)$, see \cite{y-epg4p-89}.
Hence, $b_0 \not\in V(B)$ if
 $B= b_0, b_1,\ldots, b_q$  with $q \geq 0$,  in general.
Vertices   $b_1$  and    the \emph{first} and $b_q$   the \emph{last} vertex of $B$
and edges  $\edge{b_0}{b_1}$  and $\edge{b_0}{b_q}$  are the \emph{first} and
the \emph{last} edge of $B$, respectively.
The least outer vertex in the face containing  the last edge $\edge{b_0}{b_p}$ of block $B$
 is called the \emph{dominator} of $B$, denoted $\alpha(B)$ if $B$ has at least two vertices.
Vertex $\alpha(B)$ \emph{sees} $B$ according to     \cite{y-epg4p-89}.
 Similarly, there is a \emph{last outer vertex}
$\omega(B)$, which is the  least outer vertex in the face containing
the first edge $\edge{b_0}{b_1}$.
Note that there may be edges $\edge{b_0}{v}$ with outer vertices $v$  such that
  $v < \alpha(B)$ and  $v > \omega(B)$, respectively.
If $B$ is the root of block-tree  $\mathcal{T}$, then its first outer vertex
$\alpha(\mathcal{T})$ does not necessarily see $B$. Then the first
vertex of $\mathcal{T}$, that is $\lambda(\mathcal{T})$, is an
\emph{elementary block} and the root of $\mathcal{T}$.
Recall that $\lambda(\mathcal{T})$ is connected to the last outer vertex
of $\mathcal{T}$ by an edge.

Yannakakis \cite{y-epg4p-89} assumes
 that $\alpha(B) < V(B) < \omega(B)$ for any block $B$.
We need a generalization, since there are inner separation pairs and multi-edges.
We say that  block $B$ is \emph{covered} by the outer vertex $v$ if
$v=\alpha(B)=\omega(B)$. By biconnectivity, if
there is a binding edge between a vertex of $B$
and an outer vertex if $B$ is an extreme block of $\mathcal{T}$, that is it
has no parent. Hence,
 any binding edge is incident to $v$ if $B$ is covered by $B$ and the binding
is incident to a vertex of $V(B)$. In consequence, there is a close
relationship between covered blocks and separation pairs.\\

 The following observation by Yannakakis \cite{y-epg4p-89}, also
stated in \cite{bbkr-book1p-17},  describes the
structure of the inner subgraph.

\begin{lemma}     \label{lem:partition}
Let $H$ be a 2-connected planar 2-level multigraph with outer cycle
$C = v_0,\ldots, v_t$. Then the
following hold.
\begin{enumerate}[(i)]
  \item  Any block $B$ has a first and a last outer vertex
$\alpha(B)$ and $\omega(B)$, such that $\alpha(B)\leq \omega(B)$.
  \item  Block $B$
   is in an inner component at an inner separation pair
$\seppair{\lambda(B)}{\omega(B)}$   if
and only if   $\alpha(B)=\omega(B)$, that is $B$ is covered by $v$.
  \item If $B$ is an uncovered block with leader $b_0$, $\alpha(B)=v_i$, $\omega(B)=v_j$
  and $v_i \neq v_j$, then $H-\{v_i, b_0,v_j\}$ partitions into a right part
$H_1$ and a left part $H_2$, such that $H_2$ contains the vertices $v_{i+1},
\ldots, v_{j-1}$ and the vertices of $B$.  $H_1$ is the other part.
$H- \{v_i, v_j\}$ partitions similarly
with the vertices $v_{i+1},
\ldots, v_{j-1}$ in the left part $H_2$ if $\edge{v_i}{v_j}$ is an
outer chord.
\end{enumerate}
\end{lemma}

\begin{proof}
Every face has an outer vertex and thus a first and last  outer
vertex. The vertices $\alpha(B)$ and $\omega(B)$  of block $B$ are
in the face containing the last and the first  edge of $B$,
respectively, if $B$ is nonelementary. We have
$\alpha(B) \leq \omega(B)$ since the outer cycle and blocks are
traversed in opposite directions.
If $B$ is elementary, then $\alpha(B)$ and $\omega(B)$ are
taken from the first face of the block-tree containing $B$, such
that $\alpha(B)<\omega(B)$.

For (ii), if $v= \alpha(B)=\omega(B)$ for some block $B$, then there is
a single outer vertex  that can see all vertices and edges of $B$.
There are faces with $v$  and the first and last edge of $B$, respectively.
Hence, $\seppair{\lambda(B)}{v}$ is a separation pair.
  Conversely, $v=\alpha(B)=\omega(B)$ if
  $\seppair{b}{v}$ is an inner separation pair such that $b$ is the
  leader of $B$ and $B$ is in an inner component.

For (iii),
if $B$ is not the root of a  block tree, then its leader
 $b_0$ is a cutvertex of the inner subgraph  that is
partitioned by the removal of $b_0$. Similarly, the outer cycle $C$
is partitioned by the removal of $v_i=\alpha(B)$ and
$v_j=\omega(B)$. There is a  curve $\Gamma$ from $v_i$ via $b_0$ to
$v_j$ that partitions the planar drawing of $H$. The curve first follows a binding edge
incident to $v_i$, which must exist  since $v_i=\alpha(B)$. Then it
goes along the boundary of the face containing $v_i$ and $b_0$,
which is the boundary of blocks.  It passes $b_0$ and then follows
the boundary of the face containing $b_0$ and $v_j$. The boundary
consists of edges from blocks and a final binding edge incident to
$v_j$, which exists, since $v_j=\omega(B)$.  We close $\Gamma)$
in the outer face.
 There is a shortcut for $\Gamma$ using edges
  $\edge{b_0}{v_i}$ and
$\edge{b_0}{v_j}$ that can be added uncrossed in the respective
faces of $H$. Now part $H_2$ is in the interior of $\Gamma$ and $H_1$ is outside.
  Similarly, there is a partition of
$H-\{v_i, v_j\}$  if $B$ is a root of a block tree or if
$\edge{v_i}{v_j}$ is an outer chord, which proves (iii).
\qed
\end{proof}

We say that a partition as in Lemma~\ref{lem:partition}(iii) is
\emph{induced} by an uncovered block $B$ or   an outer chord
$\edge{v_i}{v_j}$.
An inner separation pair  $\seppair{b}{v}$   is \emph{maximal} if the inner
vertex $b$ is in a block that is not covered by $v$.  Then there is
no inner vertex $a$ such that $\seppair{a}{v}$ is a separation pair whose inner
components contain the inner components of $\seppair{b}{v}$.
The inner components form a block-subtree with leader $b$.
The set of blocks of the inner components a maximal inner separation pair is called
a   \emph{super-block} and is denotedvby  $B^+$. The cutvertex $b$ is the
leader of $B^+$ and is not assigned to it, similar to blocks.
We order the blocks either
clockwise  or counterclockwise at $b$, depending on the later use.
Then the boundary of $B^+$ is traversed in clockwise  or counterclockwise order.
The first and last edge of $B^+$ are defined accordingly.
Clearly, any block of $B^+$ is covered by $v$, so that
$v=\alpha(B^+)=\omega(B^+)$, that is any super-block is \emph{covered}.
Also, any binding edge incident to a vertex of $B^+$  is incident to $v$.
There is not necessarily a
binding edge $\edge{c}{v}$ if $c$ is a cut-vertex in $B^+$ or $c=b$.
However, there is a binding multi-edge
$\edge{c}{v}$ if faces are triangulated  from $v$.

\subsection{Vertex Ordering} \label{sect:vertex-ordering}

By the peeling technique,  the \emph{vertex ordering} or \emph{linear layout}  of
 a planar  multigraph $G$  is composed of the vertex
ordering of its 2-level subgraphs.
 Yannakakis \cite{y-epg4p-89} has proposed two methods for the traversal of blocks
of a planar 2-level graph $H$. Suppose the outer cycle is traversed in cw-order.
In the \emph{consecutive method}, any block is traversed
individually in counterclockwise order from its leader,
 such that its vertices are consecutive at this moment, except for the cutvertex,
that is in the parent block, in general.
 Blocks in the same block-tree and with the same dominator
are visited in cw-order, whereas block-trees are visited in ccw-order.
  In the \emph{nested method},  the set of uncovered blocks of  a block-tree
with the same dominator is traversed by depth-first search \cite{clrs-ia-01}. Each
  vertex is listed exactly once at its first appearance.
The traversed blocks   form a block-subtree.
The nested method partitions the set of vertices of
a block into many segments between two cutvertices for children, such
hat each  segment can be assigned to an interval
in the vertex ordering. In particular, there is an interval for the vertices of any
block-subtree.

 We use the consecutive method at uncovered blocks
and super-blocks, as it admits a simpler description, whereas the nested method
must be used at block-expansions, as it   treats a set of blocks like a single one.
If $B$ is an uncovered block with dominator $v_i$, then traverse $B$ in
ccw-order from its leader and place it to the right of $v_i$. Blocks of the same
block-tree are ordered clockwise at $v_i$ if they are dominated by $v_i$.
If $B^+$ is a super-block at a maximal inner separation pair $\seppair{a_i}{v}$
such that $v < \omega(\mathcal{T})$,
  use copies of $v$ after $v$ such that each block of $B^+$ is dominated by a
copy of its own.   Then lay out the vertices of
the blocks of $B^+$ as before and remove the copies of $v$
(or keep them as placeholders, which are isolated vertices).
%

The \emph{block-expansion} of an uncovered  block
or super-block $A$  at its vertex $a_i$ by a super-block
$B^+$ is obtained traversing the boundary of $B^+$ in postorder \cite{clrs-ia-01},
that is any block is traversed in ccw-order, blocks with the same cutvertex
are visited in ccw-order, and the cutvertex is listed last, after the vertices
of the blocks of the block-subtree. The obtained sequence of vertices is
inserted right before $a_i$.  
The   \emph{block-expansion} of $A$ is obtained by expanding it
at any of its vertices by super-blocks that are covered by
the last outer vertex $\omega(\mathcal{T})$, and
 is denoted by $A^*$. The block-expansion of a
planar 2-level multigraph $H$  is obtained by expanding all uncovered blocks
   in  all block-trees.

The boundary of an expanded block
no longer a simple cycle, which does not matter for our further investigations.
Note that there are one or  two edges incident to $a_i$ from the first block
of any inner component at $\seppair{a_i}{v}$, such that there are several edges
between $a_i$ and vertices  of $B^+$.
In fact, there is a similarity between a block-expansion
and  an elementary root.
We treat an expanded block like an ordinary one and
let $\alpha(A^*)= \alpha(A), \omega(A^*)= \omega(A)$ and
$\lambda(A^*)= \lambda(A)$,
where $\omega(A)= \omega(\mathcal{T})$.
We use block-expansions to capture the case of ``small faces'' in
\cite{bdggmr-benpsf-20,bdggmr-benpsf-20-a, bdggmr-benpsf-22-b}.
It leads to a new vertex ordering.

For the  \emph{vertex ordering} $L(H)$ of a planar 2-level multi-graph $H$,
we first compute all super-blocks and all  block-expansions.
Blocks that are contained in any super-block $B^+$ or block-expansion
$A^*$ are discarded for a moment.
The remaining blocks,  super-blocks, and block-expansions,
 are called \emph{X-blocks}.
Now we use Yannakakis \cite{y-epg4p-89} consecutive method for X-blocks.
We obtain the vertex ordering as in \cite{y-epg4p-89} with the consecutive method
if  there are no covered blocks.
As a remainder, choose a vertex $v_0$ in
the outer face of $G$, which is set to be the least outer vertex.
Then traverse the outer cycle $C=v_0,\ldots, v_t$ from $v_0$ in
clockwise order (cw-order), such that the vertex ordering is $v_0 < \ldots
< v_t$. Blocks and expanded blocks are traversed counterclockwise (ccw-order).
 The roles of cw-order and ccw-order switch form level to level.
 By induction on the levels, let $C$ be the outer cycle of a planar 2-level
 multigraph $H$. For any outer vertex $v_i$, place the X-blocks dominated
by $v_i$ just right of $v_i$, where $X$-blocks from the same block-tree
are ordered clockwise and X-blocks from different block-trees are
ordered counterclockwise at $v_i$. In addition, if the dominator of a block
 is the leader of the outer cycle, that is $v_0$,  then place
the vertices immediately to the left of $v_1$, as in \cite{y-epg4p-89}.

The ordering of $C$ implies an ordering for the inner vertices,
blocks, X-blocks, block-trees, outer chords, and faces of $H$ and an orientation of the
edges according to the ordering of its vertices. Each face $f$ has
outer vertices and thus a \emph{first outer vertex} $\alpha(f)$  and
a \emph{last outer vertex} $\omega(f)$, which are the least and last
outer vertex in the boundary of $f$. Clearly, $\alpha(f)=\omega(f)$ is possible.
 Faces are ordered
according to their first outer vertex and in ccw-order if faces have
the same first outer vertex. For the computation of the
vertex ordering, a triangulation of each face from its first outer
vertex may be helpful, as the ordering of X-blocks
that are dominates by $v_i$
coincides with the ordering of the incident
edges and triangulation  edges, where there may be more multi-edges,
 for example, if there
  are inner separation pairs.  As an example, consider Figure~\ref{fig:2-level-graph}.

\begin{lemma} \label{lem:ordering}
Let $H$ be a 2-connected planar 2-level multigraph with outer cycle
$C = v_0,\ldots, v_t$.
If $H$ is partitioned into parts $H_1$ and $H_2$ induced by  an
uncovered  X-block
$B$ with $v_i= \alpha(B) \neq  \omega(B)=v_j$
or by an outer chord $\edge{v_i}{v_j}$,  as described in Lemma~\ref{lem:partition}, then the
vertex ordering satisfies  \hspace{5mm} $V_l < v_i < U  < V_2    < v_j   <  V_r$,  \hspace{5mm}
where $V_2$ is the set of vertices of part $H_2$, $U$ is the
set of vertices of X-blocks dominated by $v_i$ in part $H_1$, $V_l$
is the set of vertices of X-blocks dominated by vertices $v < v_i$
and $V_r$ is the set of vertices of X-blocks dominated by vertices $v \geq  v_j$.

If  $B$  is a covered X-block, then it is   a super-block that is covered by
some outer vertex $v_i$ if  $v_i < \omega(\mathcal{T})$. Now
  $L(H)$ satisfies
$V_l < v_i < U  < V_2  <  V_r$, where $V_l, U$ and $V_r$ are
as before,  and $V_2$ is the set of vertices of $B$.
\end{lemma}

\begin{proof}
The statement extends  Lemmas 1 and 2 in \cite{y-epg4p-89}, which
prove the partition   of the set of vertices and the vertex ordering
 in case of   a connected inner subgraph and no block
expansions, see also \cite{bbkr-book1p-17}.

We proceed by induction on the dominators of X-blocks and
the first vertex of an  outer chord.
All outer vertices $v < v_i$ and the X-blocks dominated by $v$ precede
$v_i$ in the vertex ordering. This includes vertices from blocks that are merged
into another block by a block-expansion. Similarly, all outer vertices
$v' > v_j$ and the X-blocks dominated by $v'$ succeed $v_j$.
So vertices of blocks that are covered by the last outer vertex of a block-tree
may move to $V_l$.
 Consider the X-blocks dominated by $v_i$.
 If   $B_1$ and $B_2$  are dominated by $v_i$ such that  $B_1 \in H_1$ and $B_2 \in H_2$,
then $B_1$ precedes $B_2$ in the vertex ordering, since $B_1$ is in the
  right and $B_2$ is in the left part. In fact, if $B_1$
and $B_2$ are in the same block tree, then there is a path from
$B_2$ to $B_1$ in their block tree, as shown by Yannakakis
\cite{y-epg4p-89} if there are no blocks that are covered by $v_i$.
  If $B_1$ and $B_2$ are in different
block trees, then the block trees are separated by an
outer chord or a face that can contain an outer chord in its interior. Now vertex $v_i$
dominates the root of the block tree $\mathcal{T}_2$ containing
$B_2$, such that $v_i < w$ for any vertex $w$ in $\mathcal{T}_2$.
 The block tree $\mathcal{T}_1$ containing $B_1$ precedes
 $\mathcal{T}_2$, such that blocks from $\mathcal{T}_1$ precede those
 from $\mathcal{T}_2$ if they are dominated by $v_i$. Hence,
 $V_l < v_i < U < V_2 < v_j$, where  $V_2 < v_j $ is clear from
\cite{y-epg4p-89}.  The case
with an outer chord $\edge{v_i}{v_j}$ is similar.
Any block $B$ that is dominated by the last outer vertex $\omega(\mathcal{T})$
of its block-tree $\mathcal{T}$
is covered by $\omega(\mathcal{T})$,
since $\omega(B) \leq \omega(\mathcal{T})$. It is
a priori merged into an extended block.
Hence, $\omega(\mathcal{T})$ does not dominate X-blocks of $\mathcal{T}$ .
The blocks from an inner component  at a separation pair $\seppair{a}{v_i}$
are merged into a single super-block $B^+$ that is placed to the right of $v_i$
and to the left of $v_{i+1}$.  There is an interval exactly for
the vertices of $B^+$. By the vertex ordering at $v_i$, X-blocks
are ordered clockwise at $v_i$ if they are dominated by $v_i$, such that $U$
contains the vertices of all X-blocks that precede $B^+$.
the inner components are ordered like block-trees and the vertices from any single
inner component are placed to the right of $v_i$ and to the left of $v_{i+1}$ .
Hence, the stated properties hold.
  \qed
\end{proof}

We now return to the original set of blocks.
If $B$ is an uncovered block with $\omega(B) = v_j$ or
$e=\edge{v_i}{v_j}$ is an outer chord, then  there are no
vertices of part $H_2$ to the right of $v_j$. If $B$ is covered by $v_i$,
then the vertices of $B$ are in an interval that is exclusive for $B$.
The interval is placed
between  $v_i$ and   $ v_{i+1}$ if $v_i \neq \omega(\mathcal{T})$ for the
block-tree containing $B$ and to the left of
the cutvertex $a_i$ of block $A$ if $v_i = \omega(\mathcal{T})$,
such that there is a block-expansion for $A$. In the latter case, $B$ is
part of a super-block that is traversed in postorder, similar to   the nested method.

\subsection{Embedding of Edges}  \label{sect:frame-edges}

Yannakakis~\cite{y-epg4p-89}  has used three pages for the embedding of
planar 2-level graphs. All outer edges, all outer chords and all
backward binding edges (between vertices of a block and its
dominator) are embedded in page $P_1$. The inner edges of a block,
in particular, the first and the last edge, are embedded in a single
page $P_2$ or $P_3$, and the forward binding edges are embedded in
the other page. These pages alternate between a block and its
parent, that is at an odd and an even distance between a block and
the root of its block-tree. We adopt this embedding after a triangulation
from the first outer vertex of each face.  Any triangulation  edge
is a crossed edge of the given maximal  2-level  framed multigraph.
We obtain a planar 2-level multi-graph $H^+$
with a set edges $E^+$ that includes the edges of $H$.
For any face $f$, let $E^-(f)$ be the set of \emph{remaining crossed edges}
and let  $E^- = \cup_f \,E^-(f)$.

Clearly, the vertex ordering $L(H)$ coincides with the vertex ordering
of $L(H^+)$, since vertex $v$ dominates a block in $L(H)$
if and only if $v$ is the least outer vertex in a face containing
$v$ and the last edge of  a non-elementary block if and only if
there is a triangle in $H^+$ with $v$ and the last edge of the block, and
similarly for the first face of a block-tree or any block-expansion.
First, we show that the edges of $H^+$ can be  embedded in
three pages using our vertex ordering $L(H)$.
Then one vertex per face is done, since its incident edges
are embedded. All these edges  can be disregarded subsequently,
which is important, in particular for Lemma~\ref{lem:noconflict-ret}.
Hence, only $k-1$
vertices remain for a $k$-face. If face $f$ contains vertices of a
covered block, then it contains a single outer vertex, such that
 $V^-(f)$ consists only of inner vertices. This even simplifies the situation.
 There is a
face-conflict graph that represents a possible conflict between two
remaining edges of any two faces. We show that the face-conflict graph
is outerplanar, such that it is 3-colorable. Moreover, it  represents
conflicts, such that there is no crossing of remaining edges from
two faces with the same color. The arguments are similar to the planar case.
In total, the edges of a 2-level $k$-framed
multigraph can be embedded in $3 \lfloor (k-1)/2 \rfloor  +3$ pages.

\begin{lemma} \label{lem:edges-firstoutervertex}
Any triangulated  planar 2-level $k$-framed  multigraph $H^+$
 can be embedded in three pages if the vertex ordering
$L(H)$   is used.
\end{lemma}
\begin{proof}
Yannakakis \cite{y-epg4p-89} has proved that all outer edges, all outer
chords and all backward binding edges of a triangulated planar
2-level graph can be embedded in page $P_1$.
The edges of block $B$ are embedded in page $P_2$ and forward
binding edges incident vertices of $B$ in page $P_3$.
Pages $P_2$ and $P_3$ alternate for $A$ if block $A$ is the parent of $B$.
Yannakakis   excludes outer chords and covered blocks.
The extension is proved in the same way using X-blocks.

If $B^+$ is a super-block that is covered by some
vertex $v_i < \omega(\mathcal{T})$ for the block-tree $\mathcal{T}$
containing $B^+$, then its vertices are placed in an
interval  $[B^+]$ to the right of $v_i$ that exclusively contains the vertices of $B^+$.
For any vertex $b$ in $B^+$,  $\edge{b}{v}$ is a backward binding edge
and is embedded in page $P_1$. The edges form a fan at $v_i$, such that they
do not twist mutually. They do not twist other edges in $P_1$,
since all vertices between $v_i$ and  (the left boundary of) $[B^+]$  are from X-blocks that
are dominated by $v_i$. If X-block
$A^+$ contains the leader   of $B^+$, then embed
the inner edges of the blocks of $B^+$   in page  $P_3$ if the forward binding
edges incident to vertices of $A^+$ are embedded in $P_3$.

Suppose the uncovered block $A$ is expanded at its vertex $a_i$ by $B^+$.
 Then there are forward binding edges $\edge{a_i}{\omega(\mathcal{T})}$
and $\edge{b}{\omega(\mathcal{T})}$ for any vertex $b$ in $B^+$,
  since $\omega(\mathcal{T})$ is the last outer vertex
that may be incident to such binding edges.
These edges are  embedded in
pages $P_2$ or $P_3$,  opposite to the page for
 the inner edges of $A$  \cite{y-epg4p-89}.
The inner edges of $B^+$ are embedded in the same page as  the inner edges of
$A$, since the nested method is used.
 The vertices of $B^+$
are immediately to the left of $a_i$, and  there  are no other vertices
in $[B^+, a_i]$, such that binding  edges with a vertex
in $[B^+, a_i]$ do not twist mutually. In particular, if $a_i$
is a cut-vertex and the leader of blocks that succeed $A$,
then the edges incident $a_i$ and these blocks can be embedded in the
page opposite to the the page used for the inner edges of $A$.
If $A$ is expanded at vertices $a_i$ and $a_j$ with $a_i<a_j$, then all binding edges
incident to a vertex $a$ of $A$ with $a_i< a < a_j$ are incident to
$\omega(\mathcal{T})$. They are all embedded in the same page.
If blocks $A$ and $B$ are expanded at vertices $a_i$ and $b_j$,  respectively,
then the binding edges are embedded by the above rule.
Suppose that blocks $A$ and $B$ are uncovered such that $B$ is the first child of
$A=a_0,\ldots, a_p$ in ccw-order. Then vertex
 $a_m$ with $0\leq m \leq q$ is the leader
of $B$ and $a_m$ is minimal. Now all vertices with an edge
incident to a vertex of $B$ are in the interval $[a_m,\omega(\mathcal{T})]$.
Block $A$ is expanded only at vertices $a_i$ with $a_0 \leq a_i \leq a_m$.
If $e$ is an edge from a block-expansion at $A$, then $e=\edge{u}{\omega(\mathcal{T})}$
with $u \leq a_m]$, whereas the vertices of any edge $e\rq{}$ incident to a vertex of $B$
is in   $[a_m,\omega(\mathcal{T})$, such that $e$ and $e\rq{}$
can be embedded in the same page.
 By induction, we obtain that
   the edges from all block-expansions, all uncovered blocks
and all super-blocks 
can be embedded in pages $P_2$ and $P_3$ without creating a conflict.
 Hence, all edges of $H^+$ can be embedded in three pages.
  \qed
\end{proof}
%

We now consider the sets of remaining edges $E^-(f)$ of the faces.
There is
no need to distinguish covered and uncovered blocks, since the edges
incident to the dominator are in the set $E_{\alpha}$.
Hence, all backward binding edges between a block and its dominator
are disregarded. In particular, if $e$ is  remaining edge with a vertex in
a covered block, then both vertices  are  inner vertices,
that are in the interior of specified interval.

\begin{definition} \label{def:bad}
A face $f$ of   a planar 2-level  multigraph is called   \emph{bad} for
block $B$ if $f$ has degree at least four   and

\noindent (i) $B$ is non-elementary and $f$ contains the last
edge of $B$  in its boundary and  (a) either $B$ is uncovered or
(b) $B$ is a covered block in a super-block  or

\noindent (ii)(a)   $B$ is a covered block that is merged into an expanded block
and  $f$ also contains the last edge of $B$ and the successor of the leader of
$B$ in its block or

\noindent (b) $B$ is  an elementary root of  a block-tree $\mathcal{T}$ and
$f$ is   the first face of  $\mathcal{T}$ that  also contains
an outer vertex $v$ with $\alpha(\mathcal{T}) < v < \omega(\mathcal{T})$
or an inner vertex.

  A face is \emph{bad} if it is bad for any block $B$, and \emph{good},
otherwise.
\end{definition}

For an  example, see Figure~\ref{fig:2-level-graph}.
 In particular, face $f$ is good if it does not contain any cutvertex in its boundary.
Clearly, a face can be bad for several blocks, namely if  they have  the same dominator.
In return, there are blocks without a bad face, for example, if there is a triangle
containing the last edge of a block.
As any non-elementary block has a last edge and any elementary one a first face,
then following is clear.

\begin{lemma}\label{lem:onebadedge}
For any   block  $B$ of a planar 2-level multigraph
there is at most one face $f$ such that $f$ is bad for $B$.
\end{lemma}

Hence, the number of bad faces is bounded by  the number of blocks.\\

Next we define a conflict between two faces via bad faces, and then
prove in Lemma~\ref{lem:noconflict-ret} that there is no conflict
between two remaining edges if the faces are not in conflict.
Note that our notion of conflict is different from the one in \cite{bdggmr-benpsf-20,bdggmr-benpsf-20-a, bdggmr-benpsf-22-b},
since we disregard all edges that are incident to the first outer vertex in any face.
This  restriction is important.

The last edge $\edge{b_0}{b_q}$ and any forward binding edge incident
to a vertex $b_i$ for $0<i<q$ of block $B=b_0,\ldots, b_q$ and $q \geq 2$
twist, as observed by Yannakakis \cite{y-epg4p-89}. Similarly,
crossed edges incident to $b_0, b_q$ and $b_i$ may twist, such that they
shall be embedded in different sets of pages.
Also the vertex of an elementary block and the cutvertex
of a bock-subtree at a block-expansion behave in the same way, since they are
spanned by crossed edges between vertices from the bad face.
For a block-tree  $\mathcal{T}$, e say that face $f$ is \emph{on the front side}   if
$v=\omega(\mathcal{T})$ for any outer vertex of $f$
 and $f$ contains an inner vertex from any   block of $\mathcal{T}$.
  The faces on the front side are ordered clockwise at $\omega(\mathcal{T})$,
the last of which contains the first vertex of $\mathcal{T}$.
 Let $A$ be an uncovered block in block-tree $\mathcal{T}$, such that $A$
is expanded at its vertex $a_i$ by some super-block $B^+$.
Let $x$ be any cutvertex of $B$ or $x=a_i$.
Then the vertices of the block-subtree of $B^+$ with leader (cutvertex) $x$
are in an interval $[B^+, x]$ immediately to the left of $x$.
Hence, the remaining vertices from any face with a vertex in the  block-subtree
are in this interval, except for the first and last face  in ccw-order at
$\omega(\mathcal{T})$. The first face is bad if it contains vertices that span $x$.
The last face contains $x$, the first vertex of the last sibling at $x$ in cw-order,
and probably the predecessor of $x$ in its block.
There are at most $2m$ faces in the front side with $x$ in their boundary if
the block-subtree has $m$ siblings at $x$.
Any of these faces may contain a crossed edge $e$  that is incident to $x$ such that
$e$ twists
any edge from the bad face that spans $x$.
If block $A$ is expanded to $A^*$, then all vertices from the block-expansion at
its vertex $a_i$ are in the interval $[a_{i-1}, a_i]$.
Hence, any remaining edge $e$ from the bad face of $A$ does not twist any remaining edge
from a face in the front side with a vertex in the expansion at a vertex of $A$,
since the vertices of $e$ are outside $[a_{i-1}, a_i]$.


\begin{definition} \label{def:conflict}
  Two faces $f$ and $f'$ of a  planar 2-level multigraph $H$ are
in \emph{conflict} if  (i) $f$ is bad  for a non-elementary  block $B$
that is uncovered or in a super-block  and  $f'$ contains a vertex of $B$
except if $f'$ is on the front side
 or (ii) $f$ is bad for a covered block or an elementary block
and $f'$ contains the leader of $B$.

  The \emph{face-conflict graph} $H^{\times}$
has the  faces of $H$ as its vertices. There is an edge
$\edge{f}{f'}$ in $H^{\times}$ if   $f$ and $f'$ are in conflict.
\end{definition}

The following is obvious.

\begin{lemma}\label{lem:faces-block}
Two faces  contain vertices of a single block including its leader if they are in conflict.
\end{lemma}

\begin{lemma}\label{lem:faces-block-twist}
Any two   edges $e \in E^-(f)$ and $e' \in E^-(f')$  do not twist
if faces  $f \neq f'$ do not both contain  vertices assigned to any block $B$.
\end{lemma}
\begin{proof}
First, assume that $f$ and $f'$ do not contain inner vertices of the same block-tree.
Then there is a partition of $H$ induced by an outer chord $\edge{v_i}{v_j}$,
that may be added by a  triangulation, as described in Lemma~\ref{lem:partition},
such that all vertices of $V^-(f)$ are in part $H_2$ and all vertices of
$V^-(f')$ are in part $H_1$. As shown in Lemma~\ref{lem:ordering}, the vertices of
$V^-(f)$ are in the interval $[v , v_j]$ with $v_i \leq v$
and those of $V^-(f)$ are outside, or vice versa.
Note that $v_i$ is disregarded if it is the first outer vertex of $f$ or $f\rq{}$.
Clearly, $e$ and $e\rq{}$ do not twist in this case.

 Next, assume that $B>B\rq{}$ in a block-tree $\mathcal{T}$. If $B$ is uncovered, then
 there is a partition
of $H$ induced by $B$ such that part $H_2$ contains $B$  and $B\rq{}$ is in
part $H_1$. As before, the vertices of $V^-(f)$ are in the interval $[v , v_j]$,
where $v_j=\omega(B)$ for the least block $B$ with vertices in $V^-(f)$,
and those of $V^-(f\rq{})$ are outside.
Similarly, if $B$ is covered by some outer vertex $v \neq  \omega(\mathcal{T})$,
then the vertices of $V(B)$ are in an interval to the right of the interval
for the vertices of $B\rq{}$. Similar to the case of forward binding edges,
the interval for the vertices of $f\rq{}$ is contained in the interval
$[\lambda(B), B]$ if $B\rq{}$ is dominated by some vertex $v\rq{} \leq v$
and $\lambda(B\rq{}) <  \lambda(B)$ are assigned to the same block,
the interval  for $V^-(f\rq{})$ precedes the one for $ V^-(f)$ if $f\rq{}<f$
and it includes the interval for $V^-(f)$ if $f\rq{} > f$.

At last, if $B$ is merged into an  expanded block $A^*$, then its vertices
are in an interval $[a_{j-1}, a_j]$, where $a_j$ is the cut-vertex in an
uncovered block $A$, such that it is disjoint from the interval for the vertices of
$B\rq{}$ and $f\rq{}$ or properly nests within that interval if $B\rq{}$
is not merged into $A^*$, too.  Now the vertices of $B$ are in an interval
to the left of $\lambda(B)$ that is either disjoint from the interval for the
vertices of $B\rq{}$, that is to the left of $\lambda(B\rq{})$, or it is
a subinterval.

Hence, in any case, there are disjoint or intervals
 for the remaining
vertices of $f$ and $f\rq{}$ or one is a subinterval of the other
, such that edges $e$ and $e\rq{}$ cannot twist.
 \qed
\end{proof}


Hence, if  $e$ and $e'$ twist and are remaining edges in two faces $f$ and $f'$,
then the  faces are \emph{close}.
In particular, edges cannot twist if their faces are separated
by an outer chord of the frame. Clearly, a face  may contain vertices
from  different block-trees and
it may be bad for several blocks, namely if the blocks have the same
dominator, and it may contain the root of several block-trees the leader of
several blocks that are merged into an expanded block, as
Figure~\ref{fig:2-level-graph} illustrates.

\begin{figure}[t]
\centering
    \includegraphics[scale=1.0]{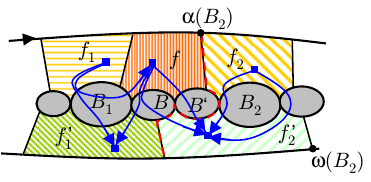}   
  \caption{
  Illustration for the proof of Lemma~\ref{lem:outerplanar}.
  A face-conflict graph (with multi-edges). Face $f_1$ is bad for $B$ and $B'$.
  Edges $\edge{f_1}{f'_1}$ and $\edge{f_2}{f'_2}$ are separated by a
  curve $\Gamma$ (red and dashed).
  }
  \label{fig:conflict-graph}
\end{figure}

\begin{lemma}\label{lem:outerplanar}
  The  face-conflict graph of a planar 2-level multigraph is outerplanar.
\end{lemma}

\begin{proof}
  The face-conflict graph $H^{\times}$ is a subgraph of the planar dual   $H$ from which
  the   outer face
  and the faces inside non-trivial blocks are removed.
  Then all  faces of $H^{\times}$ are in the outer face, since each face has
at least one outer vertex.
 A face of $H^{\times}$ is a cutvertex if it has an outer
chord in its boundary or may contain an outer chord in its interior.
It is isolated if it has no  inner vertices.

  Orient the edges of $H^{\times}$  away from  bad faces.
There are two cases.
In case (i), if face $f$ is bad for a non-elementary block $B$ that is uncovered
or in a super-block, then route an edge $\edge{f}{f'}$
 from $f$ through   $B$ to $f'$, such that any two edges
  incident to $f$, do not cross, and similarly for $f'$. Edge $\edge{f}{f'}$ enters $B$
  through the last edge $\edge{b_0}{b_q}$
  of $B=b_0,\ldots, b_q$.
  Block $B$ is entered only by edges incident
  to $f$, since $B$ has at most one  bad face. Block $B$ is left through any
  vertex $b$ assigned to $B$, except if $b$ is on the front side, that
is all faces with $b$ in their boundary are on the front side. Also
the leader of $B$ is excluded.
There may be multi-edges $\edge{f}{f'}$ if $f$ is bad for several blocks.
 Multiple copies can be removed.
In case (ii),  if $B$ is the elementary root of its block-tree $\mathcal{T}$,
or $B$ is merged into an expanded block $A^*$ and $f$ is bad for $B$,
 then route an edge $\edge{f}{f'}$ from $f$
through $b_0$ to $f'$ if $f'$ contains $b_0$ in its boundary.

  Consider two edges $\edge{f_1}{f'_1}$ and $\edge{f_2}{f'_2}$ of
  the  face-conflict graph, as illustrated in Figure~\ref{fig:conflict-graph}.
Let $f_i$ be bad for $B_i$ for $i=1,2$,  such the $B_i$ is the least
such block. First, assume that case (i) holds for both blocks.
  If the edges are adjacent, then they do not cross. In particular,
  if $f'_1=f_2$, then $B_1$ is the
  parent of $B_2$ if $B_1$ and $B_2$ are in the same block tree,
   such that $\lambda(B_2)$ is a vertex of $B_1$. Then
 $\edge{f_1}{f'_1}$ enters $f'_1$ through $\lambda(B_2)$ and $\edge{f_2}{f'_2}$
 enters
 $B_2$ through the last edge of $B_2$, such that they do not cross.
  Otherwise, assume $f_1 < f_2$.
  Then $B_1 < B_2$, since faces and blocks are ordered clockwise.
  Consider a
  curve $\Gamma$ from $\alpha(B_2)$ through $\lambda(B_2)$ to $\omega(B'_2)$,
  where $\omega(B'_2)$ is the last outer vertex in a face with vertices of $B_2$
  in its boundary. Then $\omega(B'_2)=\omega(B_2)$ if there is a
binding edge $\edge{\lambda(B_2)}{\omega(B_2)}$ in $H$. Otherwise,
$\omega(B'_2)$ is the last outer vertex after $\omega(B_2)$ in a
face that sees the first edge of $B_2$, such that $\omega(B'_2) \geq
\omega(B_2)$.
Route $\Gamma$ such that it first follows the binding edge incident
to $\alpha(B_2)$ to some block $D_1$  and then it follows the blocks
$D_1,\ldots, D_r$ on the side of the dominator $\alpha(B_2)$ to the
leader $\lambda(B_2)$ in $D_r$. There is no binding edge between
$\alpha(B_2)$ and any vertex of $D_i$ for $i>1$. Next, $\Gamma$
follows the inner boundary of $f'_2$, that is  blocks $D'_1,\ldots,
D'_s$ up to the binding edge incident to $\omega(B'_2)$ in the
boundary of $f'_2$. Then $D_r=D'_1$ and $D_i \neq D'_j$,   since
$B_2$ is the least block. Curve $\Gamma$ is routed along uncrossed
edges of the frame. It is completed to a closed curve by a part of
the outer cycle between $\alpha(B_2)$ and $\omega(B'_2)$. The faces
$f_1, f'_1$ and $f_2, f'_2$ are on opposite sides of $\Gamma$, such
that the edges $\edge{f_1}{f'_1}$ and $\edge{f_2}{f'_2}$ are on
opposite sides of $\Gamma$. Hence, the edges cannot cross.

Next, suppose that case (ii) holds for both blocks.
Then the bad face for $f_i$ contains  the leader   of $B_i$ is its boundary, $i=1,2$.
There is at most one bad face next to $b_0$, such that $b_0$ is not
passed by other edges of the conflict graph.
If $B_1$ and $B_2$ are in different block-trees or in block-expansions
such that $B_2$ is not in the block-subtree
with root $B_1$, then $\edge{f_1}{f'_1}$  and $\edge{f_2}{f'_2}$
are separated as follows.
There is  an outer chord or a binding edge incident to $\omega(\mathcal{T})$
in $H$, that separates $B_1$ and $B_2$ in $H$ as described
in Lemma~\ref{lem:partition}. Then edges $\edge{f_1}{f'_1}$  and $\edge{f_2}{f'_2}$
can meet in a common face, for example $f'_1 =f_2$, but they
cannot cross in the dual.
If $B_1$ is an ancestor of $B_2$ in a block-subtree that is part of a
block-expansion, then $f_1 <  f'_1 \leq f'_2 \leq f_2$ or
$f'_1 < f_2 \leq  f'_1 < f'_2$, such that the edges do not cross.

Consider the mixed case.
If $f_1$ is the first face of block-tree $\mathcal{T}_1$,   then
the blocks with vertices in the boundary of
$\edge{f_1}{f'_1}$  and $\edge{f_2}{f'_2}$
are separated by an outer chord if $B_2$ is not in $\mathcal{T}_1$
and by a curve through $\alpha(B_2) , \lambda(B_2), \omega(B_2)$,
as described in Lemma~\ref{lem:partition},
such that
  $\edge{f_1}{f'_1}$  and $\edge{f_2}{f'_2}$
can meet but do not cross.
In particular, if the leader of $B_2$ is the first vertex of $\mathcal{T}_1$,
that is $B_1$, then there is a triangle $(f_1, f_2, f'_2)$, where
$f'_1=f_2$ or $f'_1=f'_2$.

At last suppose that block $B_2$ is merged into an expanded block $A^*$.
If block $B_1$ is not in $A^*$, then   $\edge{f_1}{f'_1}$  and $\edge{f_2}{f'_2}$
can be separated as described before.
Let $B_1=A$ be the uncovered block that is expanded to $A^*$
and suppose that the leader of $B_2$ is vertex $a_i$ of $A$. All other cases are similar.
Then vertex $a_i$ is on the front side, such that it blocks   any edge
from the bad face of $A$. In fact, vertices $a_0,\ldots, a_m$ are blocked if they
are on the front side and $a_m$ the first vertex of $A$ that is the leader
of a block that is not covered by $\omega(\mathcal{T})$ or $a_m$ is the least
vertex of $A$ in the boundary of the bad face for $B_2$, in which case
we have $f'_1=f_2$. Otherwise,
 $\edge{f_1}{f'_1}$  and $\edge{f_2}{f'_2}$
can be separated as described before.

Hence, any two edges of the conflict graph do not cross, so that
$H^{\times}$ is outerplanar.
    \qed
\end{proof}

The next lemma includes all edges of a face. If can be restricted to the remaining
set of edges if the faces have degree at least five.

\begin{lemma} \label{lem:noconflict-ret}
There are edges $e \in E(f)$ and $e' \in E(f')$ that twist twist  if
faces $f\neq f'$ are in conflict.
\end{lemma}

\begin{proof}
Assume that $f$ is bad for block $B$. If $B=b_0,\ldots, b_q$ is
 non-elementary and is not merged into
an expanded block, then $f'$ contains a vertex $b$ of $B$ and some outer vertex $v'$
such that $\edge{b}{v'}$ is a forward binding edge. Let $u$ be a forth
vertex in $f$. Now edges $\edge{b_0}{b_q}$ and $\edge{b}{v'}$ twist.
Similarly, if $B$ is elementary or is merged into an expanded block $A^*$, then
$\lambda(B)$ is spanned by a crossed edge $\edge{a_j}{a_{j+1}}$ of $f$,
whereas the further inner vertices of $f'$ are to the left of $a_j$.
\qed
\end{proof}

The converse of Lemma~\ref{lem:noconflict-ret} is true when restricted to
the remaining edges. It completes the   correctness proof for  our algorithm.
 For the proof, we use the computed vertex ordering $L(H)$ and the fact
that edges incident to the first  outer vertex of each face are
excluded.  It resembles the case for the last edge of a block and forward and
backward binding edges incident to its vertices from Lemma~\ref{lem:edges-firstoutervertex}.

\begin{lemma} \label{lem:noconflict-ret}
Edges $e \in E^-(f)$ and $e' \in E^-(f')$ do not twist  if  faces
$f\neq f'$ are not in conflict.
\end{lemma}

\begin{figure}[t]
\centering
    \includegraphics[scale=0.8]{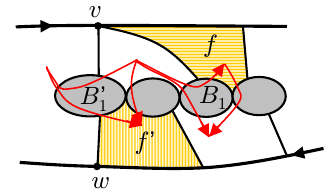}  
 \caption{
Illustration to the proof of Lemma~\ref{lem:noconflict-ret}. Faces
$f$ and $f'$ are  not in conflict and they do not
contain vertices of the same block.
  }
  \label{fig:conflict}
\end{figure}

\begin{proof}
By Lemma~\ref{lem:faces-block-twist}, $e$ and $e'$  do not twist if faces $f$ and $f'$
do not share vertices of a block $B$. If all shared  blocks are covered, then  the
remaining vertices of at least one of $f$ and $f'$ are inner vertices, which
simplifies  the situation.
As in the proof of Lemma~\ref{lem:outerplanar}, we must distinguish between
the cases from Definition~\ref{def:bad}.
Let $f < f'$ and assume that block $B$ is non-elementary and is
  uncovered or in a super-block. Then $f$ does not contain the last edge of $B$.
If $b_f$ is the least vertex of $B$ in $f$ and   $v_f$  its last outer vertex,
then
  the remaining vertices of $f$ are in the interval $[b_f, v_f]$, where $v_f$ is
the last inner vertex of $f$ if it has a single outer vertex.
Then the outer vertices of $f$ precede those of  $f'$ if $f< f'$.
The inner vertices of block $B$ in $f'$ precede those of $f$ and
inner vertices in  $f'$ succeed the last vertex of $B$, since they are
in blocks $B' > B$. Hence, edge $e$ does not twist $e'$.
Note that face $f\rq{}$ may be the first face for a block-tree with inner
vertices from block $B$ that is in a different (earlier) block-tree.

If $f$ is the first face of a block-tree $\mathcal{T}$ and $f$ and $f\rq{}$ contain
vertices of the root of $\mathcal{T}$, then the first face is not bad for $B$,
for example it is a triangle or contains only outer vertices $v \geq \omega(\mathcal{T})$
besides the first outer vertex of $\mathcal{T}$.
Then  $\mathcal{T}$ has an elementary root  $b_0$, the vertices of $V^-(f\rq{})$
are in an interval $[b_0,\omega(\mathcal{T})$ and the vertices of $V^-(f)$
are outside this interval.

At last, assume that $B$ is covered by $\omega(\mathcal{T})$, such that it is merged
into an expanded block $A^*$. Then $V^-(f)$ contains only inner vertices from
blocks that are merged into $A^*$. Then the vertices of  $V^-(f\rq{})$ are
in an interval that excludes the leader $\lambda(B)$ if $f$ is bad for $B$,
such that the intervals for  $V^-(f)$ and  $V^-(f\rq{})$ are disjoint,
or if $\lambda(B) \in  V^-(f) \cup V^-(f\rq{})$, then the intervals for
$V^-(f)$ and $V^-(f\rq{})$ share vertex $\lambda(B)$ and $[V^-(f\rq{})]$
is a subinterval of  $[V^-(f)]$ if $f < f\rq{}$. All other cases are similar.

Hence, edge $e$ does not span exactly one vertex of $e\rq{}$ or vice versa,
such that $e$ and $e\rq{}$ do not twist.
\qed
\end{proof}



We comprise the above Lemmas to the following result:

\begin{theorem}\label{thm:embed-2-levelgraph}
Any  2-level   framed multigraph $H$  can be embedded in $3 K+3$ pages
if the set of remaining edges $E^-(f)$ of  every face can be embedded
in $K$ pages using the vertex ordering $L(H)$. In particular, $H$
can be embedded in $3\lfloor k/2 \rfloor$+3 pages if $H$ is
$k$-framed.
\end{theorem}

\begin{proof}
Graph $H$ consists of a planar $k$-framed multigraph  and of sets of
crossed edges for the faces. The sets from the triangulated planar multigraph
  can be embedded in three
pages, as shown in Lemma~\ref{lem:edges-firstoutervertex}.
Then the edges of the frame and the edges incident to the first outer
vertex of each face are done.
 By
Lemma~\ref{lem:noconflict-ret}, edges from different sets $E^-(f)$
and $E^-(f')$ can be embedded in the same set of pages if $f$ and
$f'$ are not in conflict, that is     $f$ and $f'$ have the same color
in the face-conflict graph. By Lemma~\ref{lem:outerplanar}, the
face-conflict graph is outerplanar, such that it is 3-colorable. Then
three sets of $K$ pages each suffice for $E^-$. Since any edge of
$H$ is  in   $E_{\alpha}$ or $E^-$ (or both),
 all edges of $H$ are embedded in $3K+3$ pages.
 If $H$ is $k$-framed, then $K \leq \lceil (k-1)/2 \rceil$.
   \qed
\end{proof}

\begin{corollary}  \label{cor:good-faces}
The sets of crossed edges from all good faces
of a 2-level  $k$-framed multigraph
of a can be embedded in $\lceil (k-1)/2 \rceil+3$ pages.
\end{corollary}
\begin{proof}
 The face-conflict graph is discrete if all bad faces are discarded.
Then one color suffices.
   \qed
\end{proof}

\subsection{Composition} \label{sec:main}

 As observed by Yannanakis~\cite{y-epg4p-89}, the vertices of a
 block at level $\ell+1$
including the leader  are placed between two consecutive level
$\ell-1$ vertices. Here it is assumed that  vertices from blocks are
placed just right of the second outer vertex if the blocks are
dominated by the first outer vertex. The assumption is adopted from
\cite{y-epg4p-89} and has no side effects for the embedding of edges
at any level.     The  vertices of a block at level
$\ell$ are consecutive in $L(H)$, whereas there is an interval
 between $a_j$ and $a_{j+1}$
if there is a block-expansion at $a_j$. The vertices in the interval
are incident to other vertices in the interval,
  to $a_k$ and by binding edges to $\omega(\mathcal{T})$
if $\mathcal{T}$ is the block containing (the block of) $a_j$.
 Hence, it does not matter that the
vertices of a block at level $\ell+1$ are not consecutive in $L(H)$,
similar to the nested method.
In consequence, the same set of pages can be reused for all odd
(even) levels, such that twice the number of pages for 2-level
graphs suffices for a book embedding of framed multigraphs.

The page for the inner edges of block $B$ at level $\ell$ can be reused
for the embedding of
 the backward binding edges and chords in its interior at the next level
$\ell+1$, as observed by Yannakakis \cite{y-epg4p-89}
for his 5-page algorithm.

\begin{theorem}\label{thm:main}
  Any $k$-framed multigraph  can be embedded
  in $6 \lfloor k/2 \rfloor +5$ pages.
  The book embedding can be computed in linear  time in the size
  of $G$ (number of vertices and edges).
\end{theorem}
\begin{proof}
The planar frame is recursively decomposed into 2-level graphs,
which are used to compute the linear ordering. Every 2-level
subgraph of a framed multigraph has a book embedding in $3 \lfloor
k/2 \rfloor+3$ pages, as shown in Theorem~\ref{thm:embed-2-levelgraph}.
These pages can be reused for all odd levels, and another set of $3
\lfloor k/2 \rfloor+2$ pages is used for the even levels.
There are no edges between any two vertices at levels $i$ and $j$
with $|i-j| \geq 2$ in the planar frame. Since crossed edges are in the interior of faces
of the frame, this also holds for the edges of a framed multigraph.
One edge is saved, as described before  \cite{y-epg4p-89}.

 Concerning the running time, the frame of a framed
  multigraph with $n$ vertices has at most $3n-6$ edges including
   multiplicities for multi-edges,
since there are vertices on either side of a 2-cycle by a multi-edge. It thus has $O(n)$ faces.
  The vertex ordering can be computed in linear time in the number
  of vertices, both for 2-level multigraphs and the frame. Similarly, the 3-coloring
  of an outerplanar 2-level face-conflict graph is computable in linear time, such that
  the coloring of all faces takes $O(n)$ time.
   Finally, any edge of $G$ can be embedded in constant time.
    \qed
\end{proof}

\section{Applications} \label{sect:application}

For any even $k$, a $n$-clique with $n= 3k/2$ can be represented
by a $k$-map \cite{cgp-mg-02}. Then three points  of degree $k$
support all adjacencies,  for example $k=4$ for $K_6$.  Hence, we obtain an improved lower
bound on the book thickness of $k$-map graphs. In improved upper bound is obtained
from Theorems~\ref{thm:map-multiframe} and \ref{thm:main}.

\begin{theorem} \label{thm:kmapbook}
The book thickness of $k$-map graphs is at most $6 \lfloor k/2
\rfloor +5$.
\end{theorem}

\begin{corollary} \label{cor: mapbook-lower}
For any $k \geq 3$, there  are $k$-map graphs (or $k$-framed
multigraphs) with book
  thickness at least $\lfloor 3k/4 \rfloor$.
\end{corollary}
\begin{proof}
 The book thickness of $K_n$ is $\lceil n/2 \rceil$ \cite{bk-btg-79}
and  $K_n$ is a $k$-map graph  for  $n \leq \lfloor 3k/2 \rfloor$
\cite{cgp-mg-02}.
  \qed
\end{proof}

Chen et al.~\cite{cgp-rh4mg-06} have observed that any triangulated
1-planar graph is a 4-map graph. The drawing of a triangulated
1-planar graph consists of triangles and quadrangles, which contain
a pair of crossed edges, as shown by Alam et al.\cite{abk-sld3c-13}.
Hence, any triconnected 1-planar graph is a  $4$-framed graph. In
general, there are W-configurations \cite{t-rdg-88} with a pair  of
crossed edges in the outer face of a component  at a separation
pair, see Figure~\ref{fig:Wconf}. Now multi-edges come into play, such
that any 2-connected 1-planar graph is a subgraph of a 4-framed
multigraph. Clearly, each   1-planar graph  can be augmented to a
2-connected 1-planar multigraph.

From Theorem~\ref{thm:kmapbook}  we  obtain a   bound  of
17 for the book thickness of 1-planar graphs, which improves
the previous bounds  of 39 in  \cite{bbkr-book1p-17} and 29 that can
be obtained from
\cite{bdggmr-benpsf-20}. We can do even better.

\begin{theorem}\label{thm:1-planar}
  Any 1-planar graph can be embedded in eleven  pages.
\end{theorem}
\begin{proof}
Any 1-planar graph is a subgraph of a 4-framed multigraph whose
faces are triangles or quadrangles and there is   a pair of crossed
edges in each quadrangle. Consider a 2-level graph.
If $f$ is a triangle, then its edges are embedded in pages $P_1, P_2$ and
$P_3$ by Lemma~\ref{lem:edges-firstoutervertex}.
If $f$ is a quadrangle, then only one crossed edge remains
for the set $V^-(f)$. By Theorem~\ref{thm:embed-2-levelgraph}, all remaining
crossed edges can be embedded in three pages. As observed in \cite{y-epg4p-89},
one page can be saved at the composition, such that eleven
 pages suffice for the book embedding of any 1-planar
graph.
 \qed
\end{proof}

The \emph{crossed cube}, as shown in Figure~\ref{fig:crossed-cube}, is a
4-framed graph, whose frame is a (planar) cube, such that their is a
pair of crossed edges in each face. It is a 1-planar graph with 8
vertices and 24 edges. It has book thickness four, since the vertex
ordering is taken from  a Hamiltonian cycle of the frame, such that
four pages suffice, and it needs four pages as shown in
\cite{bk-btg-79}.\\

Bekos et al.~\cite{bkr-optimal2-17} have characterized optimal 2-planar
graphs and have shown that edges are uncrossed or crossed twice and
edges that are crossed twice can be grouped or caged to $K_5$ if
$n$-vertex graphs are optimal and have $5n-10$ edges, see
Figure~\ref{fig:Xdodecaeder}. In consequence, if an edge is crossed
twice, then its crossing edges are incident to a common vertex and
the vertices of these three edges form $K_5$. In consequence, an
optimal 2-planar graph is a 5-map graph  \cite{b-5maps-19}, such
that it is a 5-framed graph. Bekos et al.~\cite{bdggmr-benpsf-20}
have shown that the book thickness of optimal 2-planar graphs is 23,
which can be improved.

\begin{corollary}\label{cor:1-planar}
Any clique augmented 2-planar graph, and in particular,  any
optimal 2-planar graph, can be embedded in 17 pages.
\end{corollary}

\section{Conclusion}  \label{sec:conclusion}

We extend Yannakakis   algorithm \cite{y-epg4p-89}  on the book
embedding of planar graphs by  block expansions and generalize the
approach by Bekos et al.~\cite{bdggmr-benpsf-20} from framed graphs
to framed multigraphs.
 Multi-edges help to obtain smaller faces, which
leads to fewer pages for the book embedding. Maximal framed
multigraphs coincide with map graphs if restricted to simple graphs.
Thus we improve the upper bound  for the book thickness of $d$-map
graphs of $O(\log n)$ by Dujmovi\'{c} and Frati
\cite{df-stackqueue-18} and   $6\lceil d/2 \rceil +5 $ (claimed) by
Bekos et al.~\cite{bdggmr-benpsf-20} to $6 \lfloor d/2 \rfloor  +5$.
In particular, we show that the book thickness of
 1-planar graphs is at most eleven.

There are several other classes of beyond-planar graphs
\cite{dlm-survey-beyond-19}, such as $k$-planar, fan-planar,
fan-crossing,  1-fan-bundle, fan-crossing free, and quasi-planar
graphs, for which the book thickness has not yet been investigated
in detail. It is unlikely that they are framed multigraphs, such that
new techniques are needed for upper bounds on the book thickness.


\end{document}